\documentclass[journal]{IEEEtran}
\usepackage{cite}
\usepackage{amsmath,amssymb,amsfonts}
\usepackage{algorithmic}
\usepackage{textcomp}
\usepackage{mathtools}				
\usepackage{bm}  					
\usepackage[ruled,linesnumbered]{algorithm2e}
\usepackage{subfigure}
\usepackage{xcolor}
\usepackage{color}
\usepackage{stfloats}
\usepackage{endnotes}
\usepackage[caption=false,font=normalsize,labelfont=sf,textfont=sf]{subfig}
\usepackage{textcomp}
\usepackage{stfloats}
\usepackage{url}
\usepackage{verbatim}
\usepackage{graphicx}

\usepackage{amsthm}
\newtheorem{theorem}{Theorem}
\newtheorem{lemma}{Lemma}

\newtheorem{assumption}{Assumption}

\allowdisplaybreaks[4]

\usepackage{xcolor}

\newcommand{\T}{\top}

\begin{document}

\title{Timed-Elastic-Band Based Variable Splitting for Autonomous Trajectory Planning}

\author{Hao Zhu, Kefan Jin, Rui Gao, Jialin Wang and~C.-J. Richard Shi,~\IEEEmembership{Fellow,~IEEE}

\thanks{H. Zhu and J. Wang  are with the Institute of Brain-Inspired Circuits and Systems, and Zhangjiang Fudan International Innovation Center, Fudan University, Shanghai, China, 201203. K. Jin and R.~Gao are with the Department of Ocean and Civil Engineering, Shanghai Jiao Tong University, Shanghai, China, 200240. C.-J. Richard Shi is with the Department of Electrical and Computer Engineering, University of Washington, Seattle, USA, WA 98195 (E-mail: 20112020135@fudan.edu.cn; jinkefan@sjtu.edu.cn; rui.gao@ieee.org; jialinwang16@fudan.edu.cn; cjshi.fudan@gmail.com). } 
}

\markboth{Journal of \LaTeX\ Class Files,~Vol.~14, No.~8, August~2021}
{Shell \MakeLowercase{\textit{et al.}}: A Sample Article Using IEEEtran.cls for IEEE Journals}

\maketitle

\begin{abstract}
Existing trajectory planning methods are struggling to handle the issue of autonomous track swinging during navigation, resulting in significant errors when reaching the destination. In this article, we address autonomous trajectory planning problems, which aims at developing innovative solutions to enhance the adaptability and robustness of unmanned systems in navigating complex and dynamic environments. We first introduce the variable splitting (VS) method as a constrained optimization method to reimagine the renowned Timed-Elastic-Band (TEB) algorithm, resulting in a novel collision avoidance approach named Timed-Elastic-Band based variable splitting (TEB-VS). The proposed TEB-VS demonstrates superior navigation stability, while maintaining nearly identical resource consumption to TEB. We then analyze the convergence of the proposed TEB-VS method. To evaluate the effectiveness and efficiency of TEB-VS, extensive experiments have been conducted using TurtleBot2 in both simulated environments and real-world datasets.

\end{abstract}

\begin{IEEEkeywords}
Timed-Elastic-Band (TEB), variable splitting, non-convexity, trajectory planning, autonomous systems.
\end{IEEEkeywords}

\section{Introduction}
\label{sec:introduction}

\IEEEPARstart{W}{ith} the continuous advancement of machine learning technologies, various autonomous systems are becoming integral parts of human daily life, including autonomous driving vehicles, unmanned aerial vehicles, and robotic manipulators \cite{kang2023two,zhang2024adaptive,Gao2020autonomous}, among others. These systems contribute significantly to enhancing the efficiency of human activities, prompting extensive efforts to explore novel practical applications for autonomous technologies. 

The traditional planning module in autonomous systems adopts a hierarchical structure, after map building with Simultaneous Localization and Mapping (SLAM) \cite{tan2022reconfigurable} or Visual simultaneous localization and mapping (VSLAM) \cite{liu2023boha}, featuring both a global planner and a local planner. The global planner accumulates environmental information and provides a globally optimal path, typically at a slower frequency. In contrast, the local planner is responsible for quickly searching for an optimal trajectory that adheres to the global path. Well-known global planning methods, such as Dijkstra's algorithm~\cite{1959A} and the A* algorithm~\cite{hart1968formal}, Jump Point Search algorithm~\cite{duchovn2014path}, rapidly exploring random tree (RRT) algorithm \cite{lavalle2001rapidly}, Informed-RRT*~\cite{6942976} and probability road map algorithm~\cite{simeon2000visibility}, are commonly employed for this purpose. In terms of local planning, Dynamic Window Approach (DWA) \cite{fox1997dynamic} and Timed-Elastic-Band (TEB)~\cite{rosmann2012trajectory}  can conceptualize the path planning problem as a multi-objective optimization problem. However, it turns out that when applied to non-convex optimization problems, the methods will remain computationally expensive.

Researchers have proposed various methods to enhance the performance of trajectory planning algorithms, particularly based on TEB and DWA algorithms.  Examples of these enhanced approaches include the Improved Dynamic Window Approach (IDWA) introduced by Zhang et al. \cite{zhang2009autonomous}, Model Predictive Control with TEB (MPC-TEB) by Rosmann et al. \cite{rosmann2015timed}, and egoTEB presented by Smith et al. \cite{smith2020egoteb}. It appears that one of the challenges is the difficulty in addressing the robot swing in the path, and the existing methods are not effectively mitigating this issue. Additionally, there seems to be a concern about accumulated errors during the movement process, leading to significant deviations from the intended trajectory, especially when reaching the endpoint.

Numerous optimization methods, such as Douglas-Rachford splitting (DRS) \cite{Douglas1992}, Peaceman--Rachford splitting (PRS) \cite{Peaceman1955}, first-order primal-dual (FOPD) method~\cite{Chambolle2011fopd}, and alternating direction method of multipliers (ADMM)\cite{ Boyd2011admm}, leverage variable splitting and are well-suited for a range of constrained optimization problems. These techniques are particularly adept at transforming problems with inequality constraints into equivalent problems with equality constraints through variable splitting. In cases where inequality constraints are present, the solution of objective problems can be hard to directly compute.

\begin{figure}[!ht]
	\centering
	\includegraphics[scale=0.5,width=0.5\textwidth]{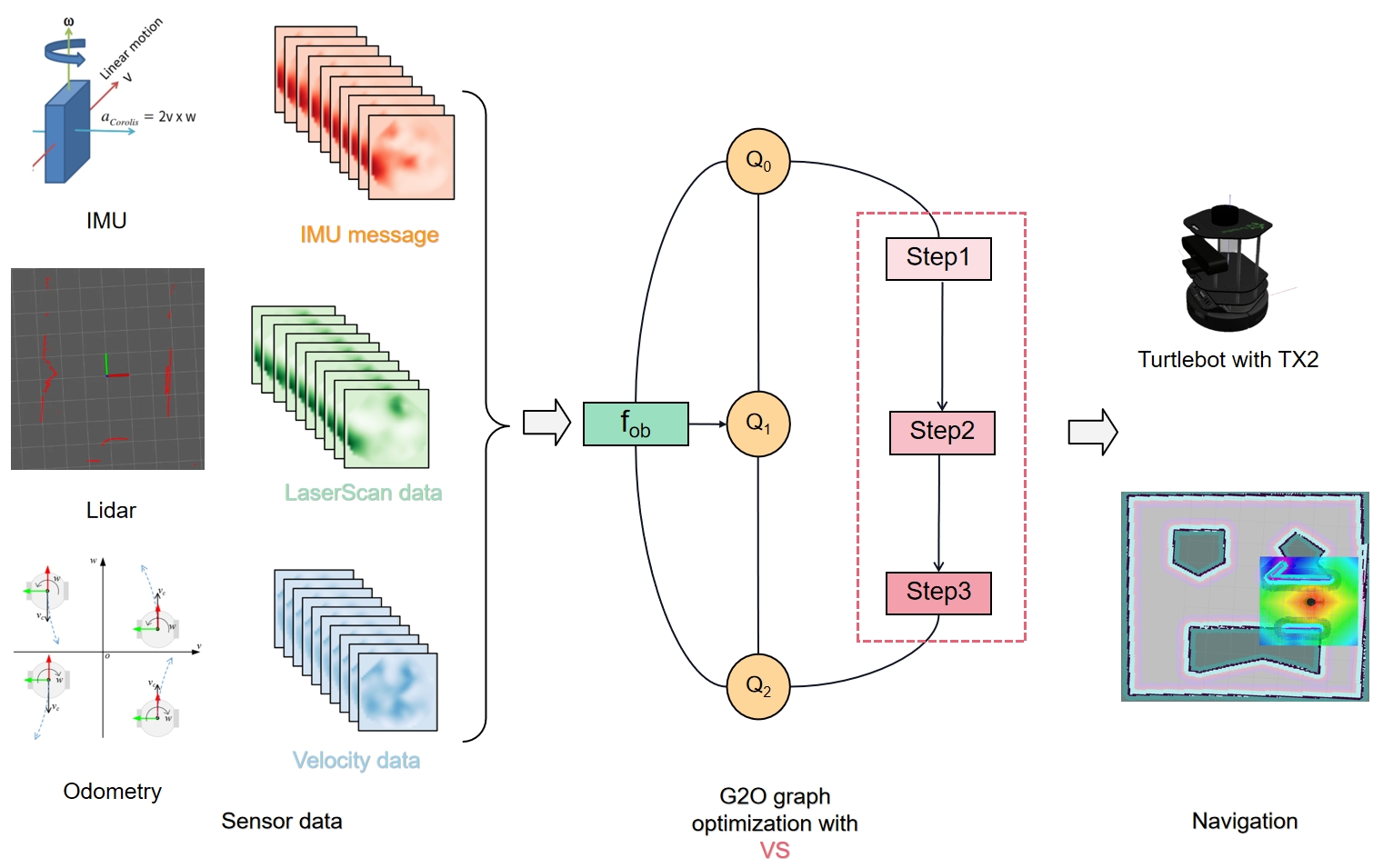}
\caption{Our architecture: i) Real-time dynamic robot state information $Q_i$ is acquired from sensors, encompassing variables such as position, orientation, and other pertinent data. ii) In the method, the objective function is introduced as a new edge in the graph, replacing the original constraint. This step is pivotal for assimilating the results of the optimization process back into the overall trajectory planning framework. Steps 1-3 correspond to formulas \eqref{eq:admm_x_linear}, \eqref{eq:admm_iteration_v}, and \eqref{eq:admm_iteration_eta}). iii) General Graph Optimization (G2O) \cite{2011G2o} is employed to solve the subproblem arising in VS framework, where each vertex in the graph model denotes as a state, and edges represent objective functions. iv) The output of TEB-VS is the optimized trajectory for the robot.}
	\label{fig:System of simulation}
\end{figure}

In this article, we present a Timed-Elastic-Band based variable splitting (TEB-VS) method for solving autonomous trajectory planning problems. TEB-VS is highly effective for decomposing the global constrained optimization problem into several small constrained optimization problems, which has the superiority of computation time efficiency. The method is designed to be highly effective in decomposing the global constrained optimization problem into a series of local small constrained optimization problems. We also establish the convergence analysis of the TEB-VS method under mild assumptions. To assess the effectiveness and efficiency of TEB-VS, we apply the TEB-VS method to real-world datasets. The overall architecture is illustrated in Fig.~\ref{fig:System of simulation}.

\section{Problem Formulation}
\label{Section:rreliminary}

We consider general TEB sequences of $I$ robot configurations as

\begin{equation} \label{eq:TEB}
	\begin{split}
		\Theta = \{\mathbf{x}_0,(\mathbf{x}_i,\Delta T_i)\},\,i=1,\ldots,I.
	\end{split}
\end{equation}

Here, $\Delta T_i$ represents the time interval, $\mathbf{x}_i = [x_i, y_i, \beta_i]^\T$ denotes the configuration of the robot pose, and $\mathbf{x}_0$ is an initial parameter. $x_i$, $y_i$, and $\beta_i$ represent the $x$-directional position, the $y$-directional position, and the orientation of the robot in the map coordinate system, respectively. In each time interval, the robot transitions from the current configuration to the next configuration in the TEB sequence $\Theta$. The primary objective is to derive the optimal TEB sequence $\Theta^\star$ as the trajectory.

Let $f_i(\Theta_i)$ be defined as the function associated with each robot configuration $\Theta_i \in \Theta$. Given a set of robot configurations $\{\mathbf{x}_i\}_{i=0}^{I}$ and corresponding time intervals $\{\Delta T_i\}_{i=1}^{I}$, our objective is to perform constrained optimization.

\begin{equation} \label{eq:general_function}
	\begin{split}
		\min_{\Theta}
		&F(\Theta) \\
		{\mathrm{s.t.}}\
		&\mathbf{c}_k \big(\mathbf{x}_0, \mathbf{x}_1, \ldots,\mathbf{x}_I \big)  = \mathbf{0},\, k=1,\ldots,K, \\
		&\mathbf{p}_j \big(\mathbf{x}_0, \mathbf{x}_1, \ldots,\mathbf{x}_I \big) \leq \mathbf{0},\, j=1,\ldots,J,
	\end{split}
\end{equation}

where $F(\Theta) = \sum \gamma_i f_i(\Theta_i)$ represents the cost function, and $\gamma_i$ serves as the weight for each function $f_i(\Theta_i)$. The parameters $K$ and $J$ denote the number of equality and inequality constraint functions, respectively, and $\mathbf{c}_k$, $\mathbf{p}_j$ represent constraint functions—such as approximating turning radius, enforcing maximum velocities, maximum accelerations, and minimum obstacle separation. 

Solving the problem \eqref{eq:general_function} presents a greater challenge compared to the conventional TEB problem, as the constraint functions in \eqref{eq:general_function} are nonconvex. In this following, we begin by establishing the framework of the proposed TEB-VS method.

\section{Methodology}
\label{Section:method}

In this section, we first introduce the framework of the proposed TEB-VS optimization method, and then derive the augmented TEB algorithm for solving the subproblems arising within the steps of the method. Under mild conditions, we also give the convergence analysis of the method.

\subsection{The General Framework}

In this section, we exploit auxiliary variables to split the complex components, and introduce a barrier function $I(\mathbf{v})$ to replace the inequality constrained function. We can rewrite \eqref{eq:general_function} as an equality constrained minimization problem

\begin{equation} \label{eq:general_function_inequality}
	\begin{split}
		\begin{aligned}
			\min_{\Theta}
			&F(\Theta), \\
			{\mathrm{s.t.}}\
			&\mathbf{c}_k \big(\mathbf{x}_0, \mathbf{x}_1, \ldots,\mathbf{x}_I \big)  = \mathbf{0}, \, k = 1,\ldots,K, \\
			&\mathbf{p}_j \big(\mathbf{x}_0, \mathbf{x}_1, \ldots,\mathbf{x}_I \big) + \mathbf{v}_j = \mathbf{0}, \, j = 1,\ldots,J.\\
			& I(\mathbf{v}_j)= \begin{cases}
				0, &   \mathbf{v} \geq 0, \\
				\infty, &     \text{otherwise}.
			\end{cases}
		\end{aligned}
	\end{split}
\end{equation}

For simplifying the notation, we define the states $\mathbf{x} =\{\mathbf{x}_1, \ldots,\mathbf{x}_I\}$, $\mathbf{v} =\{\mathbf{v}_1, \ldots,\mathbf{v}_J \}$, and Lagrangian multipliers $\bm{\eta} = \{\bm{\eta}_1,\ldots,\bm{\eta}_J\}$, $\bm{\zeta}=\{\bm{\zeta}_1,\ldots,\bm{\zeta}_K\}$.
The augmented Lagrangian function is then defined as

\begin{equation}\label{eq:Lagrangian_func}
	\begin{split}
		\begin{aligned}
			&\mathcal{L}(\mathbf{x},\mathbf{v};\bm{\eta},\bm{\zeta})  =
			\sum_{i=1}^{I} \gamma_i f_i(\Theta_i)
			+ \sum_{j=1}^{J}\bm{\eta}_j^\T( \mathbf{p}_j(\mathbf{x})  + \mathbf{v}) \\
			&+ \sum_{j=1}^{J} \frac{\rho}{2} \|  \mathbf{p}_j(\mathbf{x})  + \mathbf{v} \|^2
			+ \sum_{k=1}^{K}\bm{\zeta}_k^\T(\mathbf{c}_k(\mathbf{x}))
			+ \sum_{k=1}^{K}\frac{\gamma}{2} \|\mathbf{c}_k(\mathbf{x})\|^2,
		\end{aligned}
	\end{split}
\end{equation}

where $\rho$, $\gamma$ are penalty parameters with the constraints, respectively. At each iteration, we can write

\begin{subequations}
	\label{eq:admm_step}
	\begin{align} \label{eq:admm_x_linear}
		\mathbf{x}^{(n+1)} &= \mathop{\arg\min}_\mathbf{x}\mathcal{L}(\mathbf{x},\mathbf{v}^{(n)};\bm{\eta}^{(n)},\bm{\zeta}^{(n)}), \\
		\label{eq:admm_iteration_v}
		\mathbf{v}^{(n+1)} &= \mathop{\arg\min}_\mathbf{v}\mathcal{L}(\mathbf{x}^{(n+1)},\mathbf{v};\bm{\eta}^{(n)},\bm{\zeta}^{(n)}), \\
		\label{eq:admm_iteration_eta}
		\bm{\eta}_j^{(n+1)} & = \bm{\eta}_j^{(n)} +  \rho (\mathbf{p}_j(\mathbf{x}^{(n+1)})  + \mathbf{v}^{(n+1)}) , \\
		\bm{\zeta}_k^{(n+1)} &= \bm{\zeta}_k^{(n)} + \gamma \mathbf{c}_k(\mathbf{x}^{(n+1)}),
	\end{align}
\end{subequations}

where $n$ is the iteration number. In the objective function \eqref{eq:general_function}, acceleration and velocity optimization is one of the applications of Lagrangian Splitting. The velocity and acceleration can be minimax limits on velocity and acceleration. For the dual variable $\mathbf{v}$, the solution in \eqref{eq:admm_iteration_v} can be computed by
by 
\begin{equation}
		\begin{aligned}
			\mathbf{v}_j^{(n+1)} &= {\max} \left( \mathbf{0}, \,  -\mathbf{p}_j(\mathbf{x}^{(n+1)}) - \bm{\eta}_j^{(n)}  /\rho \right). 
			\end{aligned}
\end{equation}

The primal function \eqref{eq:admm_x_linear} is non-convex and lacks a closed-form solution. Consequently, we employ the augmented TEB method to address this challenge in the following.

\subsection{Augmented TEB for Solving $\mathbf{x}$-subproblem}

In this section, we propose the augmented TEB algorithm for solving the $\mathbf{x}$-subproblem arising in the VS framework. Typically we have the objective function

\begin{equation}\label{eq:primal_x_func}
	\begin{split}
		\begin{aligned}
			&\mathbf{x}  =
			\mathop{\arg\min}_\mathbf{x} \sum_{i=1}^{I} \gamma_i f_i(\Theta_i)\\
			&+ \sum_{j=1}^{J} \frac{\rho}{2} \left \|  \mathbf{p}_j(\mathbf{x})  + \mathbf{v} + \frac{\bm{\eta}_j}{\rho} \right\|^2 
			+ \sum_{k=1}^{K}\frac{\gamma}{2} \left\|\mathbf{c}_k(\mathbf{x})+ \frac{\bm{\zeta}_k}{\gamma} \right\|^2.
		\end{aligned}
	\end{split}
\end{equation}

The problem of optimizing $\{(\mathbf{x}_1,\Delta T_1), \ldots,(\mathbf{x}_I,\Delta T_I) \}$ can be addressed by constructing the G2O hypergraph \cite{rosmann2012trajectory,2011G2o}. Utilizing a large-scale sparse matrix, the G2O serves as the optimization model. As depicted in Fig.\ref{fig:System of simulation}, the components of time $\{\Delta T_i\}$ and configuration are treated as nodes, while objective functions act as edges. The combination of nodes and edges forms the hypergraph. In the augmented TEB method, the classical constraint edges are transformed into objective functions using variable splitting. This process involves representing the edges from traditional constraints to objective functions with variable splitting techniques. The optimized edges and vertices resulting from this transformation are then incorporated into the sparse optimizer. Subsequently, the sparse optimizer generates the corresponding results.

On the premise that the other input parameters of TEB are unchanged, replacing the speed constraint to \eqref{eq:admm_step} as edge. In a similar manner, the VS method can be applied to handle other constraints and objectives within the trajectory planning problem. For instance, if there are constraints related to waypoints, obstacles, and the fastest path, these can be reformulated and replaced using the VS method. 

\subsection{Convergence Analysis}

In this section, we establish the convergence analysis of the method. We first make the following assumptions.

\begin{assumption}
	\label{assump:prox-regular}
The function $F(\mathbf{x})$ is strongly amenable at $\mathbf{x}$ and $F(\mathbf{x})$ is \textit{prox-regular} \cite{Rockafellar1998Variational}. That is, for any $\mathbf{x}_a$, $\mathbf{x}_b $ in a neighbourhood of $\mathbf{x}$, there exists $e_{x}> 0$ such that  

\begin{equation}\label{eq:prox-regular}
	\begin{split}
		\begin{aligned}
			F(\mathbf{x}_a) \geq F(\mathbf{x}_b) - \frac{e_{x}}{2} \|\mathbf{x}_a - \mathbf{x}_b \|_2^2  +  \langle \partial F(\mathbf{x}), \mathbf{x}_a - \mathbf{x}_b\rangle.
		\end{aligned}
	\end{split}
\end{equation}

\end{assumption}
\begin{assumption}
	\label{assump:teb}
The sequence $\{\mathbf{x}^{(n)}\}$ generated by the augmented TEB converges to a local minimum $\mathbf{x}^\star$. 
\end{assumption}

We then prove that the sequence  $\{\mathcal{L}(\mathbf{x}^{(n)},\mathbf{v}^{(n)},\bm{\eta}^{(n)},\bm{\zeta}^{(n)}) \}$ generated by TEB-VS is monotonically non-increasing in the following lemma.

\begin{lemma}
	\label{lemma:nonincreasing}
	Let $\{\mathbf{x}^{(n)},\mathbf{v}^{(n)}, \bm{\eta}^{(n)},\bm{\zeta}^{(n)}\}$ be the sequence. If $\rho e_c^2 +  \gamma e_p^2 >e_x $, and the Jacobian $\mathbf{J}_c$, and  $\mathbf{J}_p$ have full-column rank with 
	
	\begin{equation}
		\begin{split}
			\mathbf{\mathbf{J}}_c^\T \mathbf{\mathbf{J}}_c \succeq e_c^2 \mathbf{I}, \,	
			\mathbf{\mathbf{J}}_p^\T \mathbf{\mathbf{J}}_p \succeq e_p^2 \mathbf{I}, 
			\quad e_c, e_p > 0,
		\end{split}
	\end{equation}
	
	Then, the sequence $\{\mathbf{x}^{(n)},\mathbf{v}^{(n)}, \bm{\eta}^{(n)},\bm{\zeta}^{(n)}\}$ is nonincreasing with the iteration number $n$. 
\end{lemma}

\begin{proof}
	For the $\bm{\zeta}$-subproblem, we obtain 
	
	\begin{equation}
		\label{eq:zeta_sub}
		\begin{split}
			\begin{aligned}
				& \mathcal{L}(\mathbf{x}^{(n+1)},\mathbf{v}^{(n+1)};\bm{\eta}^{(n+1)},\bm{\zeta}^{(n+1)}) - \\
				&\quad \mathcal{L}(\mathbf{x}^{(n+1)},\mathbf{v}^{(n+1)};\bm{\eta}^{(n+1)},\bm{\zeta}^{(n)})  \\
				&=  \frac{1}{\gamma} \| \bm{\zeta}^{(n+1)} -  \bm{\zeta}^{(n)}\|^2.
			\end{aligned}
		\end{split} 
	\end{equation} 
	
	Similar to the $\bm{\eta}$-subproblem, we obtain 
	
	\begin{equation}
		\label{eq:eta_sub}
		\begin{split}
			\begin{aligned}
				& \mathcal{L}(\mathbf{x}^{(n+1)},\! \mathbf{v}^{(n+1)};\! \bm{\eta}^{(n+1)},\! \bm{\zeta}^{(n)})\! - \! 
				\mathcal{L}(\mathbf{x}^{(n+1)},\! \mathbf{v}^{(n+1)};\! \bm{\eta}^{(n)},\! \bm{\zeta}^{(n)})  \\
				&  = \langle	\bm{\eta}^{(n+1)}, \mathbf{c}(\mathbf{x}^{(n)}) + \mathbf{v}  \rangle -
				\langle  	\bm{\eta}^{(n)}, \mathbf{c}(\mathbf{x}^{(n)}) + \mathbf{v}\rangle  \\
				&=  \frac{1}{\rho} \| \bm{\eta}^{(n+1)} -  \bm{\eta}^{(n)}\|^2.
			\end{aligned}
		\end{split} 
	\end{equation} 
	
 We set $F(\mathbf{x}) = \sum_{i=1}^{I} F(\mathbf{x}_i,\Delta T_i)$. For the $\mathbf{x}$-subproblem in \eqref{eq:admm_x_linear}, we can obtain
	
	\begin{equation}  \label{eq:x_sub}
		\begin{split}
			&  \mathcal{L}(\mathbf{x}^{(n)},\mathbf{v}^{(n)};\bm{\eta}^{(n)},\bm{\zeta}^{(n)}) - 
			\mathcal{L}(\mathbf{x}^{(n+1)},\mathbf{v}^{(n)};\bm{\eta}^{(n)},\bm{\zeta}^{(n)})  \\
			& = F (\mathbf{x}^{(n)})  
			+ \langle \bm{\eta}^{(n)},  \mathbf{v}^{(n)} + \mathbf{c}(\mathbf{x}^{(n)}) \rangle 
			+ \frac{\rho}{2} \|  \mathbf{v}^{(n)} + \mathbf{c}(\mathbf{x}^{(n+1)}) \|^2  \\
			&\, + \langle \bm{\zeta}^{(n)},  \mathbf{p}(\mathbf{x}^{(n)}) \rangle 
			+ \frac{\gamma}{2} \|  \mathbf{p}(\mathbf{x}^{(n+1)}) \|^2  
			 - \theta(\mathbf{x}^{(n+1)}) \\
			& \quad-   \langle \bm{\eta}^{(n)},  \mathbf{v}^{(n)} + \mathbf{c}(\mathbf{x}^{(n+1)})  \rangle
			- \frac{\rho}{2} \|  \mathbf{v}^{(n)} + \mathbf{c}(\mathbf{x}^{(n+1)})\|^2  \\
			& \qquad-   \langle \bm{\zeta}^{(n)}, \mathbf{p}(\mathbf{x}^{(n+1)})  \rangle
			- \frac{\gamma}{2} \|   \mathbf{p}(\mathbf{x}^{(n+1)})\|^2  
			\\
			& =  F(\mathbf{x}^{(n)})  
			+ \langle \bm{\eta}^{(n)},   \mathbf{c}(\mathbf{x}^{(n)}) -  \mathbf{c}(\mathbf{x}^{(n+1)})  \rangle 
			- F(\mathbf{x}^{(n+1)})  \\
			& \,  + \rho \langle  \mathbf{c}(\mathbf{x}^{(n+1)}) + \mathbf{v}^{(n)} ,   \mathbf{c}(\mathbf{x}^{(n+1)})  -\mathbf{c}(\mathbf{x}^{(n)})   \rangle   \\
			& \quad + \frac{\rho}{2} \| \mathbf{c}(\mathbf{x}^{(n+1)})  -\mathbf{c}(\mathbf{x}^{(n)}) \|^2  
			+ \frac{\gamma}{2} \|    \mathbf{p}(\mathbf{x}^{(n+1)}) -    \mathbf{p}(\mathbf{x}^{(n)})\|^2  \\
			& \,\quad + \gamma \langle  \mathbf{p}(\mathbf{x}^{(n+1)}),   \mathbf{p}(\mathbf{x}^{(n+1)}) -    \mathbf{p}(\mathbf{x}^{(n)})  \rangle \\
			& \qquad+ \langle \bm{\zeta}^{(n)},  \mathbf{p}(\mathbf{x}^{(n)} -   \mathbf{p}(\mathbf{x}^{(n+1)} ) \rangle  \\
			& >  F(\mathbf{x}^{(n)})  - F(\mathbf{x}^{(n+1)}) \\
			&\, + \langle  \bm{\eta}^{(n)} +  \rho (\mathbf{c}(\mathbf{x}^{(n+1)}) + \mathbf{v}^{(n)} ),   \mathbf{J}_c^\T (\mathbf{x}^{(n)} -  \mathbf{x}^{(n+1)})  \rangle \\
			&\quad +  \langle \bm{\zeta}^{(n)} +  \gamma(\mathbf{p}(\mathbf{x}^{(n+1)})), \mathbf{J}_p^\T (\mathbf{x}^{(n+1)} -   \mathbf{x}^{(n)}) \rangle 
			\\
			& \quad
			+ \frac{\rho}{2} \mathbf{J}_c^\T \mathbf{J}_c \|\mathbf{x}^{(n+1)})  -\mathbf{x}^{(n)} \|^2  
			+ \frac{\gamma}{2}  \mathbf{J}_p^\T \mathbf{J}_p \|   \mathbf{x}^{(n+1)}) -  \mathbf{x}^{(n)}\|^2  
			\\
			& {\geq}  - \frac{e_x}{2} \|\mathbf{x}^{(n)} - \mathbf{x}^{(n+1)} \|^2 + 
			+\frac{\rho e_c}{2} \|\mathbf{x}^{(n+1)}- \mathbf{x}^{(n)}\|^2 \\
			&\quad + \frac{\gamma \kappa_e}{2} \|\mathbf{x}^{(n+1)}- \mathbf{x}^{(n)}\|^2 \\
			& {\geq} 
			\left ( \frac{\rho  e_c^2  }{2} + \frac{\gamma  e_p^2  }{2} - \frac{e_x}{2}  \right)
			\| \mathbf{x}^{(n+1)}- \mathbf{x}^{(n)}\|^2,
		\end{split} 
	\end{equation}
	
	and by combining \eqref{eq:x_sub}, \eqref{eq:zeta_sub}, and \eqref{eq:eta_sub}, we get
	
	\begin{equation} \label{eq:non_increasing}
		\begin{split}
			\begin{aligned}
				& \mathcal{L}(\mathbf{x}^{(n+1)},\mathbf{v}^{(n+1)};\bm{\eta}^{(n+1)}, \bm{\zeta}^{(n+1)})  -\mathcal{L}(\mathbf{x}^{(n)},\mathbf{v}^{(n)};\bm{\eta}^{(n)},\bm{\zeta}^{(n)}) \\
				&\leq  \frac{e_x - \rho _1 e_c^2 - \gamma e_p^2 }{2} 
				\| \mathbf{x}^{(n+1)}- \mathbf{x}^{(n)}\|^2  
				+  \frac{1}{\rho} \| \bm{\eta}^{(n+1)} -  \bm{\eta}^{(n)}\|^2 \\
				&\quad +  \frac{1}{\gamma} \| \bm{\zeta}^{(n+1)} -  \bm{\zeta}^{(n)}\|^2,
			\end{aligned}
		\end{split} 
	\end{equation}
	which will be nonnegative provided that $\rho e_c^2 +  \gamma e_p^2 >e_x $.  Thus, we get the results.
\end{proof}

Based on Lemma~\ref{lemma:nonincreasing}, Assumptions \ref{assump:prox-regular} and \ref{assump:teb}, we now give the convergence analysis by TEB-VS in the following theorem.

\begin{theorem}[Convergence of TEB-VS\textbf{}]
	\label{theorem:TEB-VS}
	Let Assumptions \ref{assump:prox-regular} and \ref{assump:teb} be satisfied.
	Then, the sequence $\{\mathbf{x}^{(n)}, \mathbf{v}^{(n)}, \bm{\eta}^{(n)}, \bm{\zeta}^{(n)} \}$ generated by TEB-VS converges to a local minimum $(\mathbf{x}^\star, \mathbf{v}^\star, \bm{\eta}^\star,\bm{\zeta}^\star)$. 
\end{theorem}

\begin{proof}
	By Lemma~\ref{lemma:nonincreasing}, the sequence $\mathcal{L}(\mathbf{x}^{(n)}, \mathbf{v}^{(n)}; \bm{\eta}^{(n)},\bm{\zeta}^{(n)})$ is nonincreasing. Since $\mathcal{L}(\mathbf{x}^{(n)}, \mathbf{v}^{(n)}; \bm{\eta}^{(n)},\bm{\zeta}^{(n)})$ is upper bounded by $\mathcal{L}(\mathbf{x}^{(0)}, \mathbf{v}^{(0)}; \bm{\eta}^{(0)},\bm{\zeta}^{(0)})$, and lower bounded by $F(\mathbf{x}^{(n)})$. 
	There exists a local minimum $\mathbf{x}^{\star}$ such that the sequence ${\mathbf{x}^{(i)}}$ converges to $\mathbf{x}_{1:T}^{\star}$. The $\mathbf{v}_{1:T}$ subproblem has a unique minimum $\mathbf{v}^{\star}$ \cite{Boyd2011admm}. We then deduce that the sequence $\{\mathbf{x}^{(n)}, \mathbf{v}^{(n)}, \bm{\eta}^{(n)},\bm{\zeta}^{(n)}\}$ converges to $(\mathbf{x}^\star, \mathbf{v}^\star, \bm{\eta}^\star, \bm{\zeta}^\star)$.  We deduce the sequence $\{\mathbf{x}^{(n)}, \mathbf{v}^{(n)}, \bm{\eta}^{(n)}, \bm{\zeta}^{(n)}\}$ generated by TEB-VS is locally convergent to $(\mathbf{x}^\star, \mathbf{v}^\star, \bm{\eta}^\star,\bm{\zeta}^\star)$.
\end{proof}

\section{Experiment Results}\label{Section:exrements}
In this section, we assess the performance of the proposed algorithm through simulated data and real-life trajectory planning applications. First, we introduce the algorithm and present simulation results. Next, we illustrate the practical application of the algorithm in real robots. The results are presented and analyzed in both contexts. As depicted in Fig.\ref{fig:System of simulation}, in addition to the Turtlebot2 robot chassis, which includes odometry and an inertial measurement unit (IMU), radar sensors have been incorporated to determine obstacle distances. All relevant parameters and target equations are input into the G2O framework for graph optimization, and the optimized velocity is then output to odometry to govern the robot's movement.

\subsection{Experimental Settings}
Utilizing the Gazebo simulator on the Robot Operating System (ROS) platform, we construct a robot model that precisely mirrors the physical robot.  All simulation experiments are conducted on Ubuntu 16.04 operating system. Initially, the model world is established using Gazebo, and subsequently, the SLAM method is employed to map this model world in Fig.\ref{fig:gazebo}. Following the creation of the Octomap \cite{2013OctoMap}, three methods including DWA, TEB, and TEB-VS are individually applied for autonomous trajectory planning.

\begin{table*}[b]
	\begin{center}
		\caption{Velocity variation values generated by different methods}
		{
			\begin{tabular}[0.5\linewidth]{c|c|c|c|c} 
				\hline
				Process Velocity& Algorithm & Mean Value & Standard Deviation & Variance\\
				\hline
				& DWA    & 0.07790  & 0.05297 & 0.00281 \\
				Angular velocity variation in rotation & TEB    & {0.02431} & {0.04271} & {0.00182} \\
				& TEB-VS & 0.05018 & 0.09003 & 0.0081 \\
				\hline
				& DWA   & 0.00833 & {0.01887} &	{3.56123$\times10^{-4}$} \\
				Linear velocity variation in rotation  & TEB   & {0.00753} & 0.02420  &	$5.85836\times10^{-4}$ \\
				& TEB-VS& 0.01115 & 0.02067 &	$4.27124\times10^{-4}$ \\
				\hline
				& DWA &0.02523 & 0.02746 & $7.54212\times10^{-4}$  \\
				Angular velocity variation in linear motion & TEB &0.02620 & 0.01483 & $2.19787\times10^{-4}$  \\
				& TEB-VS& \textbf{0.01788} & \textbf{0.00988} & \textbf{9.76143$\times10^{-5}$}  \\
				\hline
				& DWA & 0.00417 & 0.01594 & $2.54073\times10^{-4}$ \\
				Linear velocity variation in linear motion  & TEB & $2.58556\times10^{-4}$ & $3.6577\times10^{-4}$ & $1.33788\times10^{-7}$ \\
				& TEB-VS& \textbf{2.33244$\times10^{-4}$} & \textbf{3.27672$\times10^{-4}$} & \textbf{1.07369$\times10^{-7}$} \\
				\hline
			\end{tabular}
			\label{table:Velocity variation of Algorithms in simulation}
		}
	\end{center}
\end{table*}

As the robot operates within a two-dimensional plane, its motion is systematically divided into two key components: linear motion and rotation. In terms of stability, particular emphasis is placed on monitoring the linear velocity and angular velocity between two consecutive configurations, whether the robot is moving in a straight line or rotating. This focus ensures a robust and stable autonomous trajectory planning process. The robot follows a trajectory from the left side of the map to the right side. Throughout this journey, its orientation undergoes a transformation from the left direction of the starting point to the right direction, providing a visual representation of the autonomous trajectory planning process. The trajectories generated by three distinct algorithms, all originating from the same starting point and concluding at the same endpoint.

\begin{figure}[ht]
	\centering
    \includegraphics[scale=0.3]{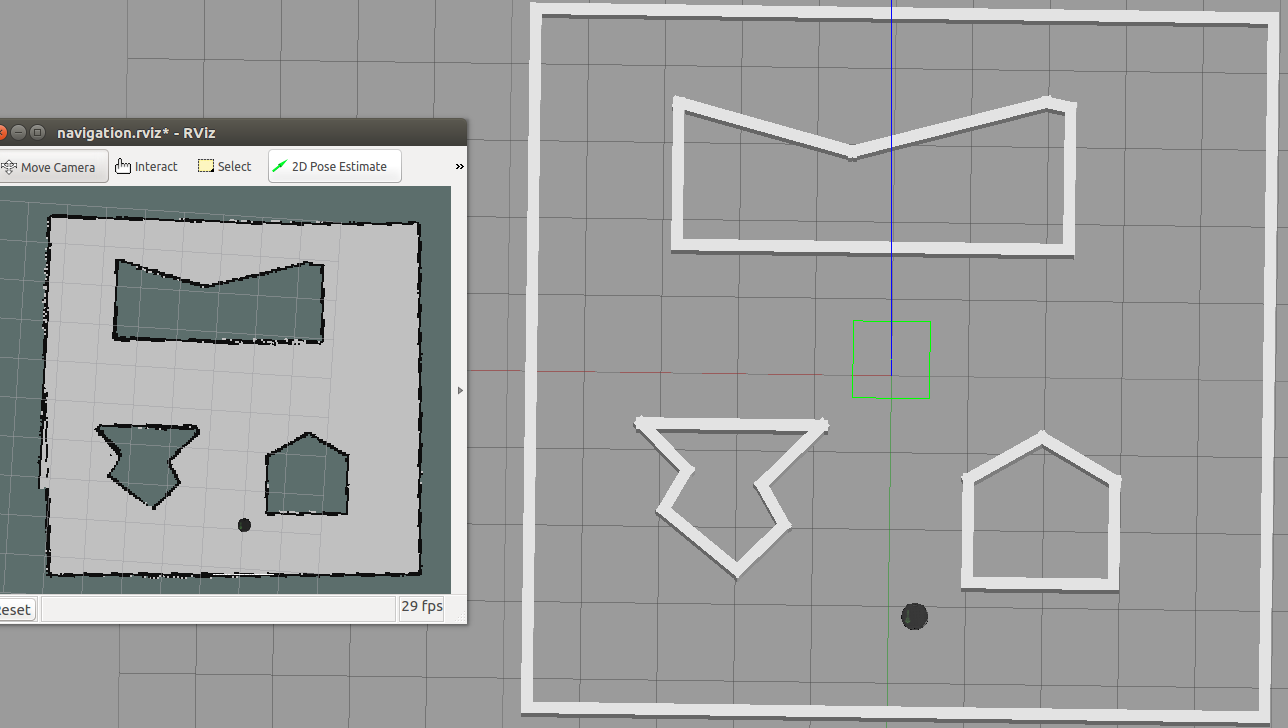}
	\caption{Gazebo world, the right part of this figure shows the simulation room by gazebo and the left part is OctoMap \cite{2013OctoMap} based on this room. The black point is robot Turtlebot2. The details will also be updated on the map in real time when robot navigation in this room.}
	\label{fig:gazebo}
\end{figure}

\begin{figure}[ht]
	\centering
	\includegraphics[scale=0.25]{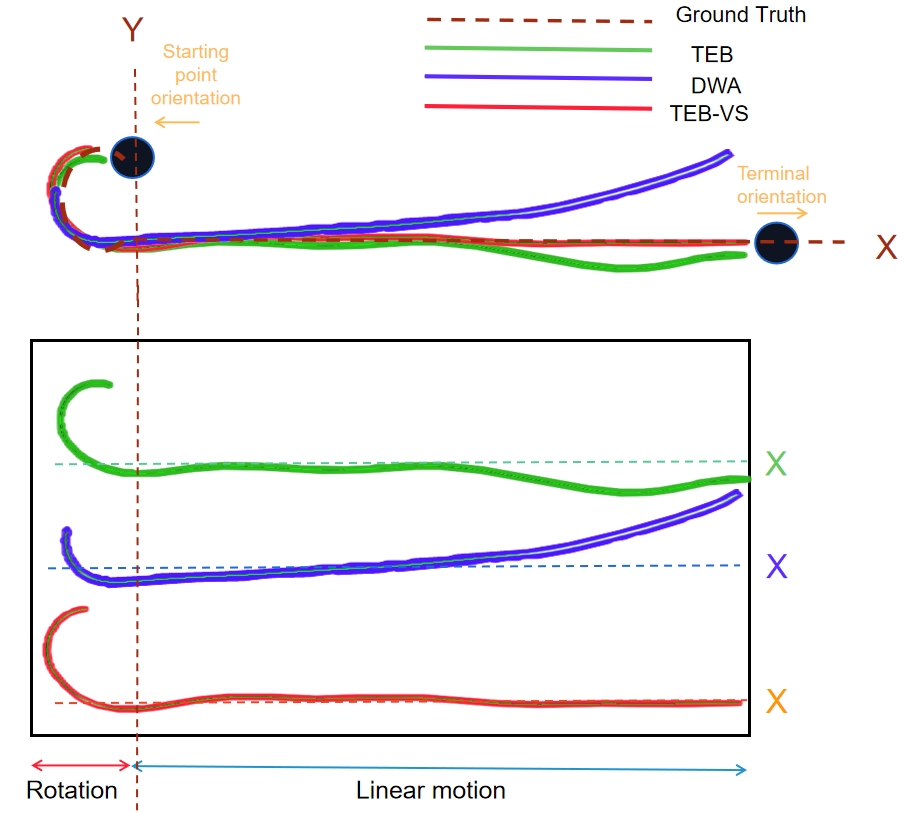}
	\caption{Autonomous trajectory planning generated by different methods. The starting point is at coordinates (2,5) with a specific orientation, and the goal is to navigate coordinates (9,5) with a different orientation in the OctoMap. There are two phases of autonomous trajectory planning motion: rotation and linear motion. The dotted lines are the ground truth. When comparing velocity stability, especially in the context of navigating the robot, it is important to distinguish between the rotational and linear processes. Special attention should be given to the left and right swing of the robot during the linear motion phase, as this significantly impact the accuracy of reaching the endpoint. }
	\label{fig:trajectory of navigation}
\end{figure}

 \subsection{Simulation Results}

We demonstrate the autonomous trajectory planning results generated by different methods in Fig.\ref{fig:trajectory of navigation}. The method outperforms the other algorithms in terms of trajectory planning accuracy, resulting in a markedly superior trajectory. Additionally, the motion of the robot following the method exhibits smoother operation throughout the autonomous trajectory planning process.

Table~\ref{table:Velocity variation of Algorithms in simulation} illustrates the velocity variation results. The comprehensive comparison of velocity stability shows that the TEB-VS method outperforms both TEB and DWA methods. During the linear motion, the comparative analysis between the TEB-VS and TEB methods reveals that the angular velocity variation in the rotation phase is notably greater for the proposed method. TEB-VS exhibits significantly better angular velocity variation compared to both TEB and DWA, achieving optimizations of over $46.5\%$ and $41.1\%$, respectively. Additionally, the standard deviation and variance of TEB-VS are smaller than those of TEB, indicating enhanced stability in linear motion.

 In addition to comparing the running trajectories, we conducted a comparative analysis of the algorithm's runtime performance and the robot's trajectory planning speed. Table~\ref{table:Algorithm run time test of simulation} presents the running time results for DWA, TEB and TEB-VS methods. Notably, given that the typical sampling frequency of the robot is generally around $1-20$Hz, the operation time for three methods remains significantly shorter than the sampling frequency. As a result, the impact on the robot's operational efficiency is minimal.

\begin{table}[!h]
	\begin{center}
		\setlength{\belowcaptionskip}{-0.1cm}
		\caption{Running time of simulation}
	{
		\begin{tabular}{c|c}
			\hline
			Method & Running time(s) \\
			\hline
			DWA & 0.00057971\\
			TEB & 0.00057971\\
			\textbf{TEB-VS} & 0.00171429\\
			\hline
		\end{tabular}
		\label{table:Algorithm run time test of simulation}
		}
	\end{center}
\end{table}

\vspace{-0.3cm}

\subsection{Real-Life Robot Results}

Similar to the simulation results, comparative tests on running time and velocity variation of the robots affirm that the method's running time is negligible, and the introduction of additional iterations has minimal impact on the algorithm's performance. In testing the actual speed change of the Turtlebot2 robot, we observed that the stability optimization result during linear motion decreased from 46.5\% in simulation to 26.7\%. Despite this reduction, the outcome remains superior to that of TEB, confirming the efficacy of the method in real-world scenarios.

\section{Conclusion}
\label{Section:conclude}
In this article, we present the new Timed-Elastic-Band based variable splitting (TEB-VS) method to solve autonomous trajectory planning problems. We also establish the convergence of the proposed method. Through extensive experiments with a real robot, our algorithm demonstrates superior effectiveness compared to the other classical methods. The TEB-VS method exhibits a 46.5\% increase in speed stability in simulation and a 26.7\% increase in real-life robot experiments. Our future focus will delve into further optimizations, particularly in the realm of fastest path or trajectory optimization.

\bibliographystyle{IEEEtran}
\bibliography{refs,refs-gr}
\end{document}